\documentclass[12pt, a4paper]{article}
\usepackage{amsmath, amssymb, amsthm}
\usepackage{mathtools}
\usepackage{graphics}
\usepackage{graphicx}
\usepackage{epsfig}
\usepackage{dblfloatfix}
\usepackage{float}
\usepackage{placeins}
\usepackage{flafter}
\usepackage[usenames]{color}

\usepackage{epstopdf}
\usepackage{url}
\usepackage{hyperref}
\date{}
\usepackage{fullpage}
\usepackage{setspace}

\newtheorem{df}{Definition}[section]

\newtheorem{thm}{Theorem}[section]
\newtheorem{cor}{Corollary}[section]
\newtheorem{lem}{Lemma}[section]
\newtheorem{exam}{Example}[section]

\makeatletter

\renewcommand\section{\@startsection {section}{1}{\z@}
{-30pt \@plus -1ex \@minus -.2ex} {2.3ex \@plus.2ex}
{\normalfont\normalsize\bfseries}}

\renewcommand\subsection{\@startsection{subsection}{2}{\z@}
{-3.25ex\@plus -1ex \@minus -.2ex} {1.5ex \@plus .2ex}
{\normalfont\normalsize\bfseries}}

\renewcommand{\@seccntformat}[1]{\csname the#1\endcsname. }

\makeatother
\title{\bf{Skew constacyclic codes over $F_{q}+uF_{q}+vF_{q}$}}
\author{ \bf Habibul Islam and Om Prakash\\\\
Department of Mathematics \\
Indian Institute of Technology Patna\\ Patna- 801 106, India \\
E-mail: habibul.pma17@iitp.ac.in and om@iitp.ac.in}

\begin{document}

\maketitle

\begin{abstract} 

In this paper skew constacyclic codes over finite non-chain ring $\mathcal{R}=F_{p^{m}}+uF_{p^{m}}+vF_{p^{m}}$, where $p$ is an odd prime and $u^{2}=u, v^{2}=v, uv=vu=0$ are studied. We show that Gray image of a skew $\alpha$-constacyclic cyclic code of length $n$ over $\mathcal{R}$ is a skew $\alpha$-quasi-cyclic code of length $3n$ over $F_{q}$ of index 3. It is also shown that skew $\alpha$-constacyclic codes are either equivalent to $\alpha$-constacyclic codes or $\alpha$-quasi-twisted codes over $\mathcal{R}$. Further, the structural properties of skew constacyclic over $\mathcal{R}$ are obtained by decomposition method.

\end{abstract}


\section{Introduction}

Codes over finite rings have been received remarkable attention for last six decades. One of the most important class of linear code is cyclic code and in more general constacyclic code. It plays crucial role in the theory of error-correcting codes and can be implemented easily by using shift registers and consequently, it is the most used codes by engineers. Earlier, linear codes were studied over finite commutative rings by researchers for a long time. But later, in 2007, Boucher et al. \cite{D07} generalized the concept of cyclic code over a non-commutative ring namely skew polynomial ring $F[x;\theta]$; where $F$ is a finite field and $\theta$ is an automorphism on the field $F$. Boucher et al. \cite{D07} and Boucher \& Ulmer \cite{D09} studied skew cyclic codes by assuming that elements of the field are not commuting  with the indeterminate $x$ and constructed some linear codes better than previously known best codes. In fact, they have considered skew cyclic codes by taking an automorphism whose order is a factor of the length of the code. By relaxing the above said restriction, Siap et al. \cite{siap11} studied the structural properties of skew cyclic codes of arbitrary length and found that a skew cyclic code is equivalent to either a cyclic code or a quasi-cyclic code over $F_{q}+vF_{q}$ where $v^{2}=v.$ In 2008, Boucher et al. \cite{D08} extended the concept of skew polynomial ring over Galois ring and studied the skew constacyclic code over Galois ring. One of the challenging task to deal with Galois ring is that the polynomial ring may not be left and right Euclidean and that insist to consider those kind of ideals which are generated by monic polynomials.
In 2012, Jitman et al. \cite{jitman} introduced skew constacyclic codes over the skew polynomial ring whose coefficients are taken from finite chain rings, particularly over the ring $F_{p^{m}} + uF_{p^{m}}$ where $u^{2} = 0.$ Very recently, in 2017, Gao et al. \cite{gao} studied skew constacyclic codes over $F_{q}+vF_{q}$ with $v^{2}=v$. Motivated by these works, in this article, we discuss skew constacyclic code over $F_{q}+uF_{q}+vF_{q}$ where $q=p^{m}$, $p$ an odd prime and $u^{2}=u, v^{2}=v, uv=vu=0.$\\

Throughout this work, $\mathcal{R}$ represents $F_{q}+uF_{q}+vF_{q}$ where $q=p^{m}$, $p$ an odd prime and $u^{2}=u, v^{2}=v, uv=vu=0.$ The ring $\mathcal{R}$ is said to be a principal ideal ring if every ideal of $\mathcal{R}$ is generated by one element. $\mathcal{R}$ is said to be a local ring if it has unique maximal ideal and chain ring if its ideals forms a chain under set inclusion.
Note that $\mathcal{R}$ is a principal ideal ring with characteristic $p$ and cardinality $q^{3}$. Also $\mathcal{R}$ is a  finite commutative semi-local and non-chain ring having three maximal ideals $\langle 1-u \rangle, \langle 1-v \rangle$ and $\langle u+v \rangle.$ \\

We consider the automorphism $\theta_{t}:F_{q}\rightarrow F_{q}$ defined by $\theta_{t}(a)=a^{p^{t}}$ and this can be extended to $\mathcal{R}$ as $\theta_{t}(a+ub+vc)=a^{p^{t}}+ub^{p^{t}}+vc^{p^{t}}$ for $a, b, c\in F_{q}.$ In this case, the order of the automorphism is $\mid \theta_{t} \mid =\frac{m}{t}=k$ (say). The invariant subring under the automorphism $\theta_{t}$ is $F_{p^{t}}+uF_{p^{t}}+vF_{p^{t}}.$ For the automorphism $\theta_{t}$ on $\mathcal{R}$ the set $\mathcal{R}[x;\theta_{t}] = \big \{ a_{0}+a_{1}x+\dots +a_{n-1}x^{n-1} \mid a_{i}\in \mathcal{R}$ for $i=1,2,3,...,n\big \}$ forms a ring under the usual addition of polynomials and the multiplication of polynomials, denoted by $\ast$, defined with respect to $ax^{i}\ast bx^{j} = a\theta_{t}^{i}(b)x^{i+j}.$ This is a non-commutative ring unless $\theta_{t}$ is identity and known as skew polynomial ring. We recall that a linear code of length $n$ is an $\mathcal{R}$-submodule of $\mathcal{R}^{n}$ and members of the code are called codewords.

\begin{df}
Let $\alpha =\alpha_{1}+u\alpha_{2}+v\alpha_{3}$ be a unit in $\mathcal{R}$ where $\alpha_{i}\in F_{p^{t}}\setminus \big \{0\big \}$ for $i=1,2,3$. A linear code $\mathcal{C}$ of length $n$ over $\mathcal{R}$ is said to be a skew $\alpha$-constacyclic code with respect to the automorphism $\theta_{t}$ if and only if $\mathcal{C}$ is invariant under the skew $\alpha$-constacyclic shift operator $\tau_{\alpha}:\mathcal{R}^{n}\rightarrow \mathcal{R}^{n}$, defined by $\tau_{\alpha}(c_{1},c_{2},...,c_{n})=(\alpha\theta_{t}(c_{n}),\theta_{t}(c_{1}),..,\theta_{t}(c_{n-1}))$, $i.e, \mathcal{C}$ is a skew $\alpha$-constacyclic code if and only if $\tau_{\alpha}(\mathcal{C})=\mathcal{C}.$
This code becomes skew cyclic code when $\alpha=1$ and skew negacyclic code when $\alpha=-1$.
\end{df}
\begin{df}
A code $\mathcal{C}$ of length $mn$ over $F_{q}$ is said to be a skew $\alpha$-quasi-cyclic code of index $m$ if ${\pi}_{m}(\mathcal{C})=\mathcal{C}$, where ${\pi}_{m}$ is the skew $\alpha$-quasi-cyclic shift on $(F_{q}^{n})^{m}$ define by
${\pi}_{m}(a^{1}\mid a^{2}\mid...\mid a^{m})=(\tau_{\alpha}(a^{1})\mid \tau_{\alpha}(a^{2})\mid...\mid \tau_{\alpha}(a^{m}))$.
\end{df}

The results of this paper have been organized as follows: In section 2, we define Gray map and describe the Gray image of skew constacyclic code over $\mathcal{R}$. Section 3 gives some important results on linear codes while section 4 discusses skew constacyclic over $\mathcal{R}$. In section 5, we study some structural properties of skew constacyclic codes over $\mathcal{R}$ and finally present two examples to support the derived results. Section 6 concludes the paper.

\section{Gray map and Gray image of skew constacyclic code over $\mathcal{R}$}

Recall that for any $a\in \mathcal{R}^{n}$ the Hamming weight $w_{H}(a)$ defined by the number of non-zero components in $a$. The Hamming distance between $a$ and $b$ in $\mathcal{R}^{n}$ defined by $d_{H}(a,b)=w_{H}(a-b)$ and Hamming distance for a code $\mathcal{C}$ is defined by $d_{H}(\mathcal{C})= min\big \{d_{H}(a,b)\mid a\neq b, \forall a,b\in \mathcal{C}\big \}$. The Lee weight of any element $r=(r_{1},r_{2},\dots, r_{n})\in \mathcal{R}^{n}$ is define by $w_{L}(r)=\sum_{i=1}^{n}w_{L}(r_{i})$ where $w_{L}(r_{i})=w_{H}(a_{i},a_{i}+b_{i},a_{i}+c_{i})$ for $r_{i}=a_{i}+ub_{i}+vc_{i}\in \mathcal{R}, i=1,2,\dots, n.$ The Lee distance for the code $\mathcal{C}$ is defined by $d_{L}(\mathcal{C})=min\big \{d_{L}(a,b)\mid a\neq b, \forall a,b\in \mathcal{C}\big \}$ where $d_{L}(a,b)$ is Lee distance between $a$ and $b$ with $d_{L}(a,b)=w_{L}(a-b).$\\

We define the Gray map $\Phi:\mathcal{R}\rightarrow F^{3}_{q}$ by $\Phi(a+ub+vc)=(a,a+b,a+c)$ where $a,b,c\in F_{q}$.
It is checked that this map is a linear and can be extended in natural way to $\mathcal{R}^{n}$ as $\Phi :\mathcal{R}^{n}\rightarrow F^{3n}_{q}$ define by $\Phi(r_{1},r_{2},\dots, r_{n})=(a_{1},a_{2},\dots, a_{n},a_{1}+b_{1},\dots, a_{n}+b_{n},a_{1}+c_{1},\dots, a_{n}+c_{n})$ where $r_{i}=a_{i}+ub_{i}+vc_{i}\in \mathcal{R}$ for $i=1, 2,\dots, n$.\\
By definition of Gray map, the following theorem can be easily proved:

\begin{thm}
The Gray map $\Phi$ is $F_{q}$-linear isometry (or distance preserving) map from $\mathcal{R}^{n}$(Lee distance) to $F^{3n}_{q}$(Hamming distance).
\end{thm}

\begin{thm}
If $\mathcal{C}$ is $[n,k,d_{L}]$ linear code over $\mathcal{R}$, then $\Phi(\mathcal{C})$ is $[3n,k,d_{L}]$ linear code over $F_{q}.$
\end{thm}

\begin{proof}
By Theorem 2.1, $\Phi$ is a linear map, so $\Phi(\mathcal{C})$ is a linear code over $F_{q}.$ Since $\Phi$ is distance preserving as well as bijection, therefore $\Phi(\mathcal{C})$ has same minimum distance and dimension as $\mathcal{C}.$
\end{proof}

\begin{lem}
Let $\Phi$ be the Gray map from $\mathcal{R}^{n}$ to $F^{3n}_{q}$, $\tau_{\alpha}$ be the skew $\alpha$-constacyclic shift and $\pi_{3}$ be the skew $\alpha$-quasi-cyclic shift operator as defined above. Then $\Phi\tau_{\alpha}=\pi_{3}\Phi.$
\end{lem}
\begin{proof}
Let $r=(r_{1},r_{2},\dots, r_{n})\in \mathcal{R}^{n}$ where $r_{i}=a_{i}+ub_{i}+vc_{i}\in \mathcal{R}$ for $i=1,2,\dots, n.$ Then
\begin{align*}
\Phi(r_{1},r_{2},\dots, r_{n})&=(a_{1},a_{2},\dots, a_{n},a_{1}+b_{1},\dots,a_{n}+b_{n},a_{1}+c_{1},\dots, a_{n}+c_{n})
\end{align*}
Applying $\pi_{3}$, we have
\begin{align*}
\pi_{3}(\Phi(r))&=(\alpha a_{n}^{p^t},a_{o}^{p^t},\dots, a_{n-1}^{p^t},\alpha a_{n}^{p^t}+\alpha b_{n}^{p^t},\dots, a_{n-1}^{p^t}+b_{n-1}^{p^t},\alpha a_{n}^{p^t}+\alpha c_{n}^{p^t},.\\&~~~~..,a_{n-1}^{p^t}+c_{n-1}^{p^t}).
\end{align*}
On the other hand,
\begin{align*}
\tau_{\alpha}(r)&=(\alpha \theta_{t}(r_{n}),\theta_{t}(r_{0}),\dots, \theta_{t}(r_{n-1}))\\
\implies \Phi(\tau_{\alpha}(r))&=(\alpha a_{n}^{p^t},a_{o}^{p^t},\dots, a_{n-1}^{p^t},\alpha a_{n}^{p^t}+\alpha b_{n}^{p^t},\dots, a_{n-1}^{p^t}+b_{n-1}^{p^t},\alpha a_{n}^{p^t}+\alpha c_{n}^{p^t},.\\&~~~~..,a_{n-1}^{p^t}+c_{n-1}^{p^t}).
\end{align*}
Therefore, $\Phi\tau_{\alpha}=\pi_{3}\Phi.$
\end{proof}

\begin{thm}
A linear code $\mathcal{C}$ of length $n$ is a skew $\alpha$-constacycilc over $\mathcal{R}$ if and only if $\Phi(\mathcal{C})$ is a skew $\alpha$-quasi-cyclic code of length $3n$ over $F_{q}$ of index 3.
\end{thm}
\begin{proof}
Let $\Phi(\mathcal{C})$ be a skew $\alpha$-quasi-cyclic code of length $3n$ over $F_{q}$ of index 3. Then $\pi_{3}(\Phi(\mathcal{C}))=\Phi(\mathcal{C}).$ By Lemma 2.1, we have $\Phi(\tau_{\alpha}(\mathcal{C}))=\Phi(\mathcal{C}) \Rightarrow \tau_{\alpha}(\mathcal{C})=\mathcal{C}$ as $\phi$ is injective. Hence $\mathcal{C}$ is a skew $\alpha$-contacyclic code over $\mathcal{R}.$\\
Conversely, suppose $\mathcal{C}$ is a skew $\alpha$-contacyclic code of length $n$ over $\mathcal{R},$ then $\tau_{\alpha}(\mathcal{C})=\mathcal{C}.$ Applying $\Phi$, we have $\Phi(\tau_{\alpha}(\mathcal{C}))=\Phi(\mathcal{C})$ and by Lemma 2.1, $\pi_{3}(\Phi(\mathcal{C}))=\Phi(\mathcal{C}).$ Hence, $\Phi(\mathcal{C})$ is a skew $\alpha$-quasi-cyclic code of length $3n$ over $F_{q}$ of index 3.

\end{proof}

\section{Linear code over $\mathcal{R}$}
Let $B_{1}, B_{2}$ and $B_{3}$ be linear codes and following \cite{dertli15}, we define
\begin{align*}
B_{1}\otimes B_{2}\otimes B_{3}=\big \{(b_{1},b_{2},b_{3})\mid b_{i}\in B_{i} ~\forall ~i=1,2,3\big \}
\end{align*}
and
\begin{align*}
B_{1}\oplus B_{2}\oplus B_{3}=\big \{b_{1}+b_{2}+b_{3}\mid b_{i}\in B_{i} ~\forall ~i=1,2,3\big \}.
\end{align*}

We observed that any element $a+ub+vc\in \mathcal{R}$ can be written as $a+ub+vc=(1-u-v)a+u(a+b)+v(a+c).$ Let $\mathcal{C}$ be a linear code of length $n$ over $\mathcal{R}$ and let 
\begin{align*}
\mathcal{C}_{1}=\big \{a\in F^{n}_{q}\mid a+ub+vc\in \mathcal{C} ~for ~some~ b,c\in F^{n}_{q}\big \};\\
\mathcal{C}_{2}=\big \{a+b\in F^{n}_{q}\mid a+ub+vc\in \mathcal{C} ~for~ some~ c\in F^{n}_{q}\big \};\\
\mathcal{C}_{3}=\big \{a+c\in F^{n}_{q}\mid a+ub+vc\in \mathcal{C} ~for~ some~ b\in F^{n}_{q}\big \}.\\
\end{align*}
Then $\mathcal{C}_{1},\mathcal{C}_{2}$ and $\mathcal{C}_{3}$ are linear codes of length $n$ over $F_{q}$ and $\mathcal{C}$ can be uniquely expressed as $\mathcal{C}=(1-u-v)\mathcal{C}_{1}\oplus u\mathcal{C}_{2}\oplus v\mathcal{C}_{3}.$\\

Now, we present some generalize results over $\mathcal{R}$ based on the results from \cite{dertli15}. These results can be proved easily by following the same procedure of proof as given in \cite{dertli15}.

\begin{thm}
If $\mathcal{C}$ be a linear code of length $n$ over $\mathcal{R}$, then $\Phi(\mathcal{C})=\mathcal{C}_{1}\otimes \mathcal{C}_{2}\otimes \mathcal{C}_{3}$ and $\mid \mathcal{C}\mid=\mid \mathcal{C}_{1}\mid \mid \mathcal{C}_{2}\mid \mid \mathcal{C}_{3}\mid =q^{3n-\sum_{i=1}^{3}f_{i}(x)}$ where $f_{1}(x), f_{2}(x) , f_{3}(x)$ are generators of $\mathcal{C}_{1}, \mathcal{C}_{2}$ and $\mathcal{C}_{3}$ respectively.
\end{thm}
\begin{cor}
If $G_{1}, G_{2}, G_{3}$ are generator matrices of $\mathcal{C}_{1}, \mathcal{C}_{2}$ and $\mathcal{C}_{3}$ respectively, then $\mathcal{C}$ has generator matrix $G$ as\\

\[
G=
  \begin{bmatrix}
    (1-u-v)G_{1} \\
    uG_{2}\\
    vG_{3}
  \end{bmatrix}.
\]
\end{cor}
\begin{cor}
Let $\mathcal{C}=(1-u-v)\mathcal{C}_{1}\oplus u\mathcal{C}_{2}\oplus v\mathcal{C}_{3}$ be a linear code of length $n$ over $\mathcal{R}$ where $\mathcal{C}_{i}$ is $[n,k_{i},d(\mathcal{C}_{i})]$ linear code over $F_{q}$ for $i=1,2,3$. Then $\Phi(\mathcal{C})$ is $[3n,k_{1}+k_{2}+k_{3},min\big \{d(\mathcal{C}_{1}),d(\mathcal{C}_{2}),d(\mathcal{C}_{2})\big \}]$ linear code over $F_{q}.$
\end{cor}

For a code $\mathcal{C}$ dual code is denoted by $\mathcal{C}^{\perp}$ and defined as $\mathcal{C}^{\perp}=\big \{ x\in \mathcal{R}^{n} \mid x.y=0 ~\forall ~y\in \mathcal{C}\big \}$ where inner product for any two codewords $x=(x_{1},x_{2},...,x_{n}), y=(y_{1},y_{2},...,y_{n})$ in $\mathcal{R}^{n}$ is define by $x.y=\sum_{i=1}^{n}x_{i}y_{i}$. The code $\mathcal{C}$ is said to be self-dual code if and only if $\mathcal{C}^{\perp}=\mathcal{C}.$
\begin{thm}
Let $\mathcal{C}$ be a linear code of length $n$ over $\mathcal{R}$. Then $\Phi(\mathcal{C}^{\perp})=\Phi(\mathcal{C})^{\perp}$. Further, $\mathcal{C}$ is self-dual if and only if $\Phi(\mathcal{C})$ is so.
\end{thm}
\begin{proof}
For any $s_{1}=a_{1}+ub_{1}+vc_{1}\in \mathcal{C}$ and $s_{2}=a_{2}+ub_{2}+vc_{2}\in \mathcal{C}^{\perp}$ where $a_{1},a_{2},b_{1},b_{2},c_{1},c_{2}\in F^{n}_{q}$, then $s_{1}.s_{2}=0$, and this implies $a_{1}a_{2}=0, b_{1}b_{2}+a_{1}b_{2}+a_{2}b_{1}=0, c_{1}c_{2}+a_{1}c_{2}+a_{2}c_{1}=0.$ Now, $\Phi(s_{1}).\Phi(s_{2})=3a_{1}a_{2}+ b_{1}b_{2}+a_{1}b_{2}+a_{2}b_{1}+c_{1}c_{2}+a_{1}c_{2}+a_{2}c_{1}=0.$ Then $\Phi(\mathcal{C}^{\perp})\subseteq \Phi(\mathcal{C})^{\perp}$. Again, by Theorem 2.2,  $\mid \Phi(\mathcal{C}^{\perp})\mid=\mid \Phi(\mathcal{C})^{\perp}\mid$, and therefore, $\Phi(\mathcal{C}^{\perp})=\Phi(\mathcal{C})^{\perp}$. \\
Also, if $\mathcal{C}$ is self-dual $\Leftrightarrow \mathcal{C}^{\perp}=\mathcal{C} \Leftrightarrow \Phi(\mathcal{C}^{\perp})=\Phi(\mathcal{C}) \Leftrightarrow \Phi(\mathcal{C})^{\perp}=\Phi(\mathcal{C}) \Leftrightarrow \Phi(\mathcal{C})$ is self-dual.
\end{proof}

\begin{thm}
 Let $\mathcal{C}=(1-u-v)\mathcal{C}_{1}\oplus u\mathcal{C}_{2}\oplus v\mathcal{C}_{3}$ be a linear code of length $n$ over $\mathcal{R}.$ Then $\mathcal{C}^{\perp}=(1-u-v)\mathcal{C}_{1}^{\perp}\oplus u\mathcal{C}_{2}^{\perp}\oplus v\mathcal{C}_{3}^{\perp}.$ Further, $\mathcal{C}$ is self-dual if and only if $\mathcal{C}_{1},\mathcal{C}_{2},\mathcal{C}_{3}$ are self-duals.
\end{thm}
\begin{proof}
Let $\mathcal{\widehat{C}}_{1}=\big \{a\in F^{n}_{q}\mid a+ub+vc\in \mathcal{C}^{\perp}$ for some $b,c\in F^{n}_{q}\big \}$,
$\mathcal{\widehat{C}}_{2}=\big \{a+b\in F^{n}_{q}\mid a+ub+vc\in \mathcal{C}^{\perp}$ for some $c\in F^{n}_{q}\big \}$,
$\mathcal{\widehat{C}}_{3}=\big \{a+c\in F^{n}_{q}\mid a+ub+vc\in \mathcal{C}^{\perp}$ for some $b\in F^{n}_{q}\big \}$.\\ Then $\mathcal{C}^{\perp}$ can be uniquely expressed as $\mathcal{C}^{\perp}=(1-u-v)\mathcal{\widehat{C}}_{1}\oplus u\mathcal{\widehat{C}}_{2}\oplus v\mathcal{\widehat{C}}_{3}.$ Obviously, $\mathcal{\widehat{C}}_{1}\subseteq \mathcal{C}_{1}^{\perp}.$ Also, for any $s\in \mathcal{C}_{1}^{\perp} \Rightarrow s.a=0 ~\forall a ~\in \mathcal{C}_{1}.$ Let $r=(1-u-v)a+ub+vc\in \mathcal{C}$ where $b,c\in F^{n}_{q}.$ Then $(1-u-v)s.r=0\Rightarrow(1-u-v)s\in \mathcal{C}^{\perp}.$ By uniqueness of $\mathcal{C}^{\perp}$, $s\in \mathcal{\widehat{C}}_{1}$. Therefore,  $\mathcal{C}_{1}^{\perp}\subseteq\mathcal{\widehat{C}}_{1}.$ Hence  $\mathcal{C}_{1}^{\perp}=\mathcal{\widehat{C}}_{1}.$ Similarly we can see that  $\mathcal{C}_{2}^{\perp}=\mathcal{\widehat{C}}_{2},\mathcal{C}_{3}^{\perp}=\mathcal{\widehat{C}}_{3}.$ Thus $\mathcal{C}^{\perp}=(1-u-v)\mathcal{C}_{1}^{\perp}\oplus u\mathcal{C}_{2}^{\perp}\oplus v\mathcal{C}_{3}^{\perp}.$\\
Let $\mathcal{C}$ be self-dual. Then $\mathcal{C}^{\perp}=\mathcal{C}\Rightarrow \Phi(\mathcal{C}^{\perp})=\Phi(\mathcal{C}) \Rightarrow \mathcal{C}_{1}^{\perp}\otimes \mathcal{C}_{2}^{\perp}\otimes \mathcal{C}_{3}^{\perp}=\mathcal{C}_{1}\otimes \mathcal{C}_{2}\otimes \mathcal{C}_{3}$ by Theorem 3.1. Therefore, $\mathcal{C}_{i}=\mathcal{C}_{i}^{\perp}~ \forall~ i=1,2,3,~  i.e,$  $\mathcal{C}_{1},\mathcal{C}_{2},\mathcal{C}_{3}$ are self duals.\\
 Conversely, if  $\mathcal{C}_{1},\mathcal{C}_{2},\mathcal{C}_{3}$ are self-duals, then $\mathcal{C}^{\perp}=\mathcal{C}$. Therefore, $\mathcal{C}$ is self-dual.
\end{proof}

\section{Skew constacyclic codes over $\mathcal{R}$}

Since $\mathcal{R}[x;\theta_{t}]$ is non-commutative, $\langle x^{n}-\alpha \rangle$ is not a two sided ideal of $\mathcal{R}[x;\theta_{t}]$ and hence $\mathcal{R}_{n}=\frac{\mathcal{R}[x;\theta_{t}]}{\langle x^{n}-\alpha \rangle}$ is not a ring in general. It can be checked that the center of $\mathcal{R}[x;\theta_{t}]$ is $Z(\mathcal{R}[x;\theta_{t}])=(F_{p^{t}}+uF_{p^{t}}+vF_{p^{t}})[x^{k}]$ and $x^{n}-\alpha \in Z(\mathcal{R}[x;\theta_{t}])$ if and only if $k$ divides $n.$ Therefore, $\mathcal{R}_{n}$ becomes a ring only when $k$ divides $n$. However, $\mathcal{R}_{n}$ is a left $\mathcal{R}[x;\theta_{t}]-$ module even $k$ does not divide $n$, where left module multiplication is define by $r(x)(f(x)+\langle x^{n}-\alpha \rangle) = r(x)f(x)+\langle x^{n}-\alpha \rangle$ for $f(x), r(x)\in \mathcal{R}[x;\theta_{t}]$.\\
To identify each codeword of length $n$ of $\mathcal{R}^{n}$ with a polynomial in $\mathcal{R}_{n}=\frac{\mathcal{R}[x;\theta_{t}]}{\langle x^{n}-\alpha \rangle}$, we define an $\mathcal{R}$-module isomorphism from $\mathcal{R}^{n}$ to $\mathcal{R}_{n}$ by

\begin{align*}
(r_{0},r_{1},...,r_{n-1}) \longmapsto r_{0}+r_{1}x+\dots + r_{n-1}x^{n-1}
\end{align*}
With this identification we can easily prove the following result:

\begin{thm}
 A code $\mathcal{C}$ of length $n$ over $\mathcal{R}$ is a skew $\alpha$-constacyclic code if and only if the corresponding polynomial representation in $\mathcal{R}_{n}$ is a left $\mathcal{R}[x;\theta_{t}]$-submodule of $\mathcal{R}_{n}.$
\end{thm}
For this section, we consider $\alpha =\alpha_{1}+u\alpha_{2}+v\alpha_{3}$ be a unit in $\mathcal{R}$ where $\alpha_{i}\in F_{p^{t}}\setminus \big \{0\big \}$  such that ${\alpha}^{2} = 1$ $(i.e, \alpha= 1, -1, 1-2u, -1+2u, 1-2v, -1+2v, 1-2u-2v, -1+2u+2v)$. To discuss some results on skew $\alpha$-constacyclic codes over $\mathcal{R}$, we consider a map\\
\begin{align*}
\rho: \frac{\mathcal{R}[x;\theta_{t}]}{\langle x^{n}-1 \rangle} \rightarrow \frac{\mathcal{R}[x;\theta_{t}]}{\langle x^{n}-\alpha \rangle}
\end{align*}\\

defined by
\begin{align*}
\rho(f(x))= f(\alpha x)
\end{align*}
Let $n$ be an odd integer. If 
\begin{align*}
f(x)&=g(x) ~mod~ (x^{n}-1)\\
\Leftrightarrow f(x)-g(x)&=h(x)\ast(x^{n}-1) ~for~ some~ h(x)\in \mathcal{R}[x;\theta_{t}]\\
\Leftrightarrow f(\alpha x)-g(\alpha x) &=h(\alpha x)\ast (\alpha^{n}x^{n}-1)\\
& =h(\alpha x)\ast (\alpha x^{n}-\alpha^{2})(as~\alpha^{n}=\alpha~for~odd~n)\\
& =\alpha h(\alpha x)\ast ( x^{n}-\alpha )\\
\Leftrightarrow f(\alpha x) &= g(\alpha x) ~mod~ (x^{n}-\alpha).
\end{align*}
This shows that $\rho$ is well defined as well as one-one. Moreover, $\rho$ is onto and left $\mathcal{R}[x;\theta_{t}]$-module homomorphism. Towards above remark, we have the following results:
\begin{thm}
If $n$ is an odd integer, then $\rho$ is a left $\mathcal{R}[x;\theta_{t}]$-module isomorphism.
\end{thm}

\begin{cor}
Let $\mathcal{C}$ be a skew cyclic code of odd length $n$ over $\mathcal{R}$, then $\rho(\mathcal{C})$ is skew $\alpha$-constacyclic code of length $n$ over $\mathcal{R}$.
\end{cor}

\begin{thm}
Let $\mathcal{C}$ be a skew $\alpha$-constacyclic code of length $n$ over $\mathcal{R}$ and $gcd(n, k) = 1$. Then $\mathcal{C}$ is a $\alpha$-constacyclic code of length $n$ over $\mathcal{R}$.
\end{thm}
\begin{proof}
Since $gcd(n, k) = 1$, so there exits two integers $L, T$ such that $Ln+Tk = 1$. Therefore, we can choose an integer $D>0$ such that $Tk= 1+Dn$.
Let $c(x)=c_{0}+c_{1}x+\dots +c_{n-1}x^{n-1}\in \mathcal{C}$. As $\mathcal{C}$ is a skew $\alpha$-constacyclic code and $x\ast c(x)$ represents skew $\alpha$-constacyclic shift of the codeword $c(x)$, so $x\ast c(x), x^{2}\ast c(x),...,x^{Tk}\ast c(x)$ are belong to $\mathcal{C}$, where 
\begin{align*}
x^{Tk}\ast c(x)& = x^{Tk}\ast (c_{0}+c_{1}x+\dots +c_{n-1}x^{n-1})\\
& = \theta_{t}^{Tk}(c_{0})x^{Tk}+ \theta_{t}^{Tk}(c_{1})x^{Tk+1}+\dots + \theta_{t}^{Tk}(c_{n-1})x^{Tk+n-1}\\
& = c_{0}x^{1+Dn}+ c_{1}x^{2+Dn}+\dots + c_{n-1}x^{Dn+n}\\
& = \alpha^{D}(c_{0}x+c_{1}x^{2}+\dots + c_{n-2}x^{n-1}+\alpha c_{n-1})~(as~in~\mathcal{R}_{n},~ x^{n} = \alpha)\\
\implies \alpha^{D}x^{Tk}\ast c(x)& = c_{0}x+c_{1}x^{2}+\dots + c_{n-2}x^{n-1}+\alpha c_{n-1}\in \mathcal{C}~(as~ \alpha^{2} = 1)
\end{align*}
This proves that $\mathcal{C}$ is a $\alpha$-constacyclic code of length $n$ over $\mathcal{R}$.

\end{proof}

\begin{cor}
Let $gcd(n, k) = 1$. If $f(x)$ is a right divisor of $x^{n}-\alpha$ in the skew polynomial ring $\mathcal{R}[x;\theta_{t}]$, then $f(x)$ is a factor of $x^{n}-\alpha$ in the polynomial ring $\mathcal{R}[x]$.
\end{cor}
\begin{df}
Let $\mathcal{C}$ be a linear code of length $n$ over $\mathcal{R}$ and $n=ml.$ Then $\mathcal{R}$ is said to be an $\alpha$-quasi-twisted code if for any  
\begin{align*}
(c_{0 ~0},c_{0 ~1},...,c_{0 ~l-1},...,c_{m-1~ 0},c_{m-1 ~1},...,c_{m-1 ~l-1})\in \mathcal{C}  
\end{align*}
implies 
\begin{align*}
(\alpha c_{m-1 ~0},\alpha c_{m-1 ~1},...,\alpha c_{m-1~ l-1},...,c_{m-2 ~0 },c_{m-2 ~1},...,c_{m-2 ~l-1})\in \mathcal{C}
\end{align*}
If $l$ be the least positive integer satisfying $n=ml$, then $\mathcal{C}$ is known as $\alpha$-quasi-twisted code of length $n$ over $\mathcal{R}.$
\end{df}

\begin{thm}
Let $\mathcal{C}$ be a skew $\alpha$-constacyclic code of length $n$ and $gcd(n,k) = l.$ Then $\mathcal{C}$ is a $\alpha$-quasi-twisted code of index $l$ over $\mathcal{R}.$
\end{thm}
\begin{proof}
  Since $gcd(n,k) = l$, there exits two integers $T$ and $D$ such that $Tk = l + Dn ; D>0.$
Let $r = (c_{0 ~0},c_{0~ 1},...,c_{0 ~l-1},...,c_{m-1~ 0},c_{m-1~ 1},...,c_{m-1~ l-1})\in \mathcal{C}$. Since $\mathcal{C}$ is a skew $\alpha$-constacyclic code, $\tau_{\alpha}(r), \tau_{\alpha}^{2}(r),...,\tau_{\alpha}^{l}(r)$ are belong to $\mathcal{C}$, where
\begin{align*}
\tau_{\alpha}^{l}(r)&= (\theta_{t}^{l}(\alpha c_{m-1~ 0}),...,\theta_{t}^{l}(\alpha c_{m-1~ l-1}),\theta_{t}^{l}(c_{0~ 0}),...\theta_{t}^{l}(c_{0 ~l-1}),...,\\& ~~~~\theta_{t}^{l}(c_{m-2~ 0}),...,\theta_{t}^{l}(c_{m-2 ~l-1}))\\
\implies \tau_{\alpha}^{l+Dn}(r)&=(\theta_{t}^{l+Dn}(c_{m-1~ 0}),...,\theta_{t}^{l+Dn}(c_{m-1~ l-1}),...,\\& ~~~~\theta_{t}^{l+Dn}(\alpha c_{m-2~ 0}),...,\theta_{t}^{l+Dn}(\alpha c_{m-2~ l-1}))\\
&= (\theta_{t}^{Tk}(c_{m-1 ~0}),...,\theta_{t}^{Tk}(c_{m-1 ~l-1}),.\\&~~~~..,\theta_{t}^{Tk}(\alpha c_{m-2~ 0}),...,\theta_{t}^{Tk}(\alpha c_{m-2~ l-1}))\\
&= (c_{m-1~ 0},...,c_{m-1~ l-1},...,\alpha c_{m-2~ 0},...,\alpha c_{m-2~ l-1})\\
\implies \alpha \tau_{\alpha}^{l+Dn}(r)&= (\alpha c_{m-1~ 0},...,\alpha c_{m-1~ l-1},...,c_{m-2 ~0},...,c_{m-2 ~l-1})\in \mathcal{C}~(as~ \alpha^{2} = 1)
\end{align*}
This proves that $\mathcal{C}$ is an $\alpha$-quasi-twisted code of index $l$.

\end{proof}

\section{Decomposition of skew constacyclic codes over $\mathcal{R}$}

In previous section, we discussed skew $\alpha$-constacyclic code assuming $\alpha^{2}=1$. In present section relaxing the above restriction we discuss some structural properties of skew $(\alpha_{1}+u\alpha_{2}+v\alpha_{3})$-constacyclic codes over $\mathcal{R}$ by the decomposition method.

\begin{lem}
An element $\alpha =\alpha_{1}+u\alpha_{2}+v\alpha_{3}\in \mathcal{R}$ where $\alpha_{i}\in F_{q}$ for $i=1,2,3$ is unit in $\mathcal{R}$ if and only if $\alpha_{1}, \alpha_{1}+\alpha_{2}, \alpha_{1}+\alpha_{3}$ are units in $F_{q}.$
\end{lem}
\begin{proof}
Let $\alpha =\alpha_{1}+u\alpha_{2}+v\alpha_{3}\in \mathcal{R}$ be a unit in $\mathcal{R}$. Then there exits an element $\beta=\beta_{1}+u\beta_{2}+v\beta_{3}\in \mathcal{R}$ such that $\alpha.\beta=1.$ So $\alpha_{1}\beta_{1}=1, \alpha_{1}\beta_{2}+\alpha_{2}\beta_{2}+\alpha_{2}\beta_{1}=0, \alpha_{1}\beta_{3}+\alpha_{3}\beta_{3}+\alpha_{3}\beta_{1}=0$ which implies that $\alpha_{1}\neq 0, \alpha_{1}+\alpha_{2}\neq 0, \alpha_{1}+\alpha_{3}\neq 0.$ \\
Conversely, let  $\alpha_{1}\neq 0, \alpha_{1}+\alpha_{2}\neq 0, \alpha_{1}+\alpha_{3}\neq 0.$ Let $\alpha^{-1}=(1-u-v)\alpha_{1}^{-1}+u(\alpha_{1}+\alpha_{2})^{-1}+v(\alpha_{1}+\alpha_{3})^{-1}.$ Then $\alpha.\alpha^{-1}=1.$ Hence $\alpha =\alpha_{1}+u\alpha_{2}+v\alpha_{3}\in \mathcal{R}$
is a unit.\\
\end{proof}

\begin{thm}
A linear code $\mathcal{C}=(1-u-v)\mathcal{C}_{1}\oplus u\mathcal{C}_{2}\oplus v\mathcal{C}_{3}$ is a skew $(\alpha_{1}+u\alpha_{2}+v\alpha_{3})$-constacyclic code of length $n$ over $\mathcal{R}$ with respect to the automorphism $\theta_{t}$ if and only if $\mathcal{C}_{1}, \mathcal{C}_{2}, \mathcal{C}_{3}$ are skew  $\alpha_{1}$-constacyclic code, skew $(\alpha_{1}+\alpha_{2})$-constacyclic code and skew $(\alpha_{1}+\alpha_{3})$-constacyclic code of length $n$ over $F_{q}$ respectively with respect to the same automorphism.
\end{thm}
\begin{proof}
Let $s=(s_{1},s_{2},...,s_{n})\in \mathcal{R}$ where $s_{i}=(1-u-v)a_{i}+ub_{i}+vc_{i}$ for $i=1,2,...,n$. Consider $a=(a_{1},a_{2},...,a_{n}), b=(b_{1},b_{2},...,b_{n}), c=(c_{1},c_{2},...,c_{n}).$ Then $a\in \mathcal{C}_{1}, b\in \mathcal{C}_{2}, c\in \mathcal{C}_{3}.$ Suppose $\mathcal{C}_{1}, \mathcal{C}_{2}$ and $\mathcal{C}_{3}$ are skew  $\alpha_{1}$-constacyclic code, skew $(\alpha_{1}+\alpha_{2})$-constacyclic code and skew $(\alpha_{1}+\alpha_{3})$-constacyclic code of length $n$ over $F_{q}$ respectively. Then $\tau_{\alpha_{1}}(a)\in \mathcal{C}_{1}, \tau_{\alpha_{1}+\alpha_{2}}(b)\in \mathcal{C}_{2}$ and $\tau_{\alpha_{1}+\alpha_{3}}(c)\in \mathcal{C}_{3}.$ Now, $\tau_{\alpha_{1}+u\alpha_{2}+v\alpha_{3}}(s)=(1-u-v)\tau_{\alpha_{1}}(a)+u\tau_{\alpha_{1}+\alpha_{2}}(b)+v\tau_{\alpha_{1}+\alpha_{3}}(c)\in (1-u-v)\mathcal{C}_{1}\oplus u\mathcal{C}_{2}\oplus v\mathcal{C}_{3}=\mathcal{C}.$ Therefore, $\mathcal{C}$ is a skew $(\alpha_{1}+u\alpha_{2}+v\alpha_{3})$-constacyclic code of length $n$ over $\mathcal{R}$.\\
Conversely, let $\mathcal{C}$ be a skew $(\alpha_{1}+u\alpha_{2}+v\alpha_{3})$-constacyclic code of length $n$ over $\mathcal{R}$. Let $a=(a_{1},a_{2},...,a_{n})\in \mathcal{C}_{1}, b=(b_{1},b_{2},...,b_{n})\in \mathcal{C}_{2}, c=(c_{1},c_{2},...,c_{n})\in \mathcal{C}_{3}.$ Take $s_{i}=(1-u-v)a_{i}+ub_{i}+vc_{i}$ for $i=1,2,...,n$. Then $s=(s_{1}s_{2},...,s_{n})\in \mathcal{C}.$ Since  $\mathcal{C}$ is a skew $(\alpha_{1}+u\alpha_{2}+v\alpha_{3})$-constacyclic code over $\mathcal{R}$, so $\tau_{\alpha_{1}+u\alpha_{2}+v\alpha_{3}}(s)\in \mathcal{C}$ where $\tau_{\alpha_{1}+u\alpha_{2}+v\alpha_{3}}(s)=(1-u-v)\tau_{\alpha_{1}}(a)+u\tau_{\alpha_{1}+\alpha_{2}}(b)+v\tau_{\alpha_{1}+\alpha_{3}}(c)$. It follows that $\tau_{\alpha_{1}}(a)\in \mathcal{C}_{1}, \tau_{\alpha_{1}+\alpha_{2}}(b)\in \mathcal{C}_{2}$ and $\tau_{\alpha_{1}+\alpha_{3}}(c)\in \mathcal{C}_{3}$. Consequently, $\mathcal{C}_{1}, \mathcal{C}_{2}, \mathcal{C}_{3}$ are skew  $\alpha_{1}$-constacyclic code, skew $(\alpha_{1}+\alpha_{2})$-constacyclic code and skew $(\alpha_{1}+\alpha_{3})$-constacyclic code of length $n$ over $F_{q}$ respectively.
\end{proof}

\begin{cor}
Let $\mathcal{C}=(1-u-v)\mathcal{C}_{1}\oplus u\mathcal{C}_{2}\oplus v\mathcal{C}_{3}$ be a skew $(\alpha_{1}+u\alpha_{2}+v\alpha_{3})$-constacyclic code of length $n$ over $\mathcal{R}$. Then dual $\mathcal{C}^{\perp}=(1-u-v)\mathcal{C}_{1}^{\perp}\oplus u\mathcal{C}_{2}^{\perp}\oplus v\mathcal{C}_{3}^{\perp}$ is a skew $(\alpha_{1}+u\alpha_{2}+v\alpha_{3})^{-1}$-constacyclic code over $\mathcal{R}$ where $\mathcal{C}_{1}^{\perp}, \mathcal{C}_{2}^{\perp}$ and $\mathcal{C}_{3}^{\perp}$ are skew $\alpha_{1}^{-1}$-constacyclic, skew $(\alpha_{1}+\alpha_{2})^{-1}$-constacyclic and skew $(\alpha_{1}+\alpha_{3})^{-1}$-constacyclic codes of length $n$ over ${F}_{q}$ provided $n$ is a multiple of the order $k$ of the automorphism $\theta_{t}.$
\end{cor}
\begin{proof}
Since $(\alpha_{1}+u\alpha_{2}+v\alpha_{3})$ is fixed by the automorphism $\theta_{t}$ and $n$ is a multiple of $k$, so by Lemma 3.1 of \cite{jitman}, we have $\mathcal{C}^{\perp}$ is a skew $(\alpha_{1}+u\alpha_{2}+v\alpha_{3})^{-1}$-constacyclic code over $\mathcal{R}$. Now, $\alpha_{1}+u\alpha_{2}+v\alpha_{3}=(1-u-v)\alpha_{1}+u(\alpha_{1}+\alpha_{2})+v(\alpha_{1}+\alpha_{3})$, it follows that $(\alpha_{1}+u\alpha_{2}+v\alpha_{3})^{-1}=(1-u-v)\alpha_{1}^{-1}+u(\alpha_{1}+\alpha_{2})^{-1}+v(\alpha_{1}+\alpha_{3})^{-1}.$ Then by Theorem 5.1, we have $\mathcal{C}_{1}^{\perp}, \mathcal{C}_{2}^{\perp}$ and $\mathcal{C}_{3}^{\perp}$ are skew $\alpha_{1}^{-1}$-constacyclic, skew $(\alpha_{1}+\alpha_{2})^{-1}$-constacyclic and skew $(\alpha_{1}+\alpha_{3})^{-1}$-constacyclic codes of length $n$ over ${F}_{q}$.
\end{proof}

\begin{cor}
Let $\mathcal{C}=(1-u-v)\mathcal{C}_{1}\oplus u\mathcal{C}_{2}\oplus v\mathcal{C}_{3}$ be a skew $(\alpha_{1}+u\alpha_{2}+v\alpha_{3})$-constacyclic code of length $n$ over $\mathcal{R}$. Then $\mathcal{C}$ is a self-dual if and only if $(\alpha_{1}+u\alpha_{2}+v\alpha_{3}) = 1, -1, 1-2u, -1+2u, 1-2v, -1+2v, 1-2u-2v, -1+2u+2v.$
\end{cor}
\begin{proof}
The result follows from the fact that $\mathcal{C}$ is self-dual if and only if $\alpha_{1}=\pm 1, \alpha_{1}+\alpha_{2}=\pm 1, \alpha_{1}+\alpha_{3}=\pm 1.$
\end{proof}

\begin{thm}
Let $\mathcal{C}=(1-u-v)\mathcal{C}_{1}\oplus u\mathcal{C}_{2}\oplus v\mathcal{C}_{3}$ be a skew $(\alpha_{1}+u\alpha_{2}+v\alpha_{3})$-constacyclic code of length $n$ over $\mathcal{R}$. Then $\mathcal{C}$ has a generating polynomial $f(x)$ in $\mathcal{R}[x;\theta_{t}]$ such that f(x) is a right divisor of $x^{n}-(\alpha_{1}+u\alpha_{2}+v\alpha_{3}).$
\end{thm}
\begin{proof}
Let $f_{i}(x)$ be generator polynomial of $\mathcal{C}_{i}$ for $i=1,2,3$ in $F_{q}[x;\theta_{t}].$ Then $(1-u-v)f_{1}(x), uf_{2}(x)$ and $vf_{3}(x)$ are generators of $\mathcal{C}.$ Let $f(x)=(1-u-v)f_{1}(x)+uf_{2}(x)+vf_{3}(x)$ and $\mathcal{G}=\langle f(x) \rangle$. Clearly, $\mathcal{G}\subseteq \mathcal{C}.$ On the other hand, $(1-u-v)f(x)=(1-u-v)f_{1}(x)\in \mathcal{G}, uf(x)=uf_{2}(x)\in \mathcal{G}, vf(x)=vf_{3}(x)\in \mathcal{G}$. Therefore, $\mathcal{C}\subseteq \mathcal{G}$ and hence $\mathcal{C}=\mathcal{G}=\langle f(x) \rangle.$\\
Since $f_{1}(x), f_{2}(x)$ and $f_{3}(x)$ are right divisors of $x^{n}-\alpha_{1}, x^{n}-(\alpha_{1}+\alpha_{2})$ and $x^{n}-(\alpha_{1}+\alpha_{3})$ respectively, so there exist $h_{1}(x), h_{2}(x)$ and $h_{3}(x)$ such that $x^{n}-\alpha_{1}=h_{1}(x) \ast f_{1}(x); x^{n}-(\alpha_{1}+\alpha_{2})=h_{2}(x) \ast f_{2}(x); x^{n}-(\alpha_{1}+\alpha_{3})=h_{3}(x) \ast f_{3}(x).$ Now, $[(1-u-v)h_{1}(x)+uh_{2}(x)+vh_{3}(x)] \ast f(x)= (1-u-v)h_{1}(x) \ast f_{1}(x)+uh_{2}(x) \ast f_{2}(x)+vh_{3}(x) \ast f_{3}(x)= x^{n}-(\alpha_{1}+u\alpha_{2}+v\alpha_{3}).$ Therefore, $f(x)$ is a right divisor of $x^{n}-(\alpha_{1}+u\alpha_{2}+v\alpha_{3}).$
\end{proof}

\begin{cor}
Each left submodule of $\mathcal{R}[x;\theta_{t}]/\langle x^{n}-(\alpha_{1}+u\alpha_{2}+v\alpha_{3})\rangle$ is generated by one element where $\alpha_{1}+u\alpha_{2}+v\alpha_{3}$ is a unit in $\mathcal{R}$.
\end{cor}

\begin{thm}
$\mathcal{C}=(1-u-v)\mathcal{C}_{1}\oplus u\mathcal{C}_{2}\oplus v\mathcal{C}_{3}$ be a skew $(\alpha_{1}+u\alpha_{2}+v\alpha_{3})$-constacyclic code of length $n$ over $\mathcal{R}$ and $gcd(n,k)=1, gcd(n,q)=1.$ Then there exist an idempotent generator $e(x)= (1-u-v)e_{1}(x)+ue_{2}(x)+ve_{3}(x)$ in $\mathcal{R}[x;\theta_{t}]/\langle x^{n}-(\alpha_{1}+u\alpha_{2}+v\alpha_{3}) \rangle$ where $e_{1}(x), e_{2}(x), e_{3}(x)$ are idempotent generators of $\mathcal{C}_{1}, \mathcal{C}_{2}$ and $\mathcal{C}_{3}$ respectively.
\end{thm}
\begin{proof}
Let $e_{1}(x), e_{2}(x), e_{3}(x)$ are idempotent generators of $\mathcal{C}_{1}, \mathcal{C}_{2}$ and $\mathcal{C}_{3}$ respectively. Then by Theorem 5.2, we conclude that $e(x)= (1-u-v)e_{1}(x)+ue_{2}(x)+ve_{3}(x)$ is an idempotent generator of $\mathcal{C}$ in $\mathcal{R}[x;\theta_{t}]/\langle x^{n}-(\alpha_{1}+u\alpha_{2}+v\alpha_{3}) \rangle$.
\end{proof}

\begin{exam}
Consider the field $F_{9}=F_{3}[\alpha]$ where $\alpha^{2}+1=0$. Take the Frobineous automorphism $\theta_{t}:F_{9}\rightarrow F_{9}$ define by $\theta_{t}(\alpha)=\alpha^{3}$ and $\mathcal{R}= F_{9}+uF_{9}+vF_{9}$ where $u^{2}=u, v^{2}=v, uv=vu=0.$ The polynomial $f(x) = x^{6}+(1-2u-2v)x^{5}+x^{4}+(1-2u-2v)x^{3}+x^{2}+(1-2u-2v)x+1$ is a right divisor of $x^{7}-(1-2u-2v)$ in $\mathcal{R}[x;\theta_{t}]$ and $gcd(n , \mid\langle \theta_{t} \rangle\mid ) = 1$. Thus by Theorem 4.3, $\mathcal{C}=\langle f(x) \rangle$ is a $(1-2u-2v)$-constacyclic cyclic code of length 7 over $\mathcal{R}$ and also $\Phi(\mathcal{C})$ is a $[21,3,7]$ linear code over $F_{9}$.
\end{exam}

\begin{exam}
Consider the field $F_{25}=F_{5}[\alpha]$ where $\alpha^{2}+\alpha+1=0$. Take the Frobineous automorphism $\theta_{t}:F_{25}\rightarrow F_{25}$ define by $\theta_{t}(\alpha)=\alpha^{5}.$ Now, $x^{6}-1=(x^{2}-1)(x^{2}+x+1)(x^{2}-x+1)$; $x^{6}+1=(x^{2}-4)(x^{4}+4x^{2}+1).$ Let $f_{1}(x)=(x^{4}+4x^{2}+1), f_{2}(x)=(x^{2}+x+1), f_{3}(x)=(x^{2}-x+1)$ and  $f(x)=(1-u-v)f_{1}(x)+uf_{2}(x)+vf_{3}(x)=(1-u-v)x^{4}+(4-3u-3v)x^{2}+(u-v)x+1$, then $\mathcal{C}=\langle f(x) \rangle$ is a skew $(-1+2u+2v)$-constacyclic cyclic code of length 6 over $\mathcal{R}= F_{25}+uF_{25}+vF_{25}$ where $u^{2}=u, v^{2}=v, uv=vu=0.$ Since $gcd(n , \mid\langle \theta_{t} \rangle\mid ) = 2$, then by Theorem 4.4, $\mathcal{C}$ is a $(-1+2u+2v)$-quasi-twisted code of index 2 over $\mathcal{R}$. Also $\Phi(\mathcal{C})$ is a $[18,10,3]$ linear code over $F_{25}$.
\end{exam}
\section{Conclusion}
The skew $\alpha$-constacyclic codes of length $n$ over $\mathcal{R} = F_{p^{m}}+uF_{p^{m}}+vF_{p^{m}}$ is equivalent to $\alpha$-constacyclic codes  over $\mathcal{R}$ (for $gcd(n ,k) = 1$) or equivalent to $\alpha$-quasi-twisted code of index $l$ (for $gcd(n ,k) = l$). It is shown that skew constacyclic codes over $\mathcal{R}$ are principally generated. Also, we have given the necessary and sufficient conditions for codes being self-dual over $\mathcal{R}$.
\\The works what we have done in this paper can be generalized over the finite ring $F_{p^{m}}+u_{1}F_{p^{m}}+u_{2}F_{p^{m}}+\dots +u_{s}F_{p^{m}}, u_{i}^{2} = u_{i}, u_{i}u_{j} = u_{j}u_{i} = 0, 1\leq i, j \leq s$, $p$ is prime and $s\geq 1.$

\section*{Acknowledgement}
The authors are thankful to University Grant Commission(UGC), Govt. of India for financial support under Ref. No. 20/12/2015(ii)EU-V dated 31/08/2016 and Indian Institute of Technology Patna for providing the research facilities. The authors would like to thank the anonymous referees for their useful comments and suggestions.

\end{document}